\documentclass[11pt]{amsart}

\usepackage{amssymb,latexsym,amsmath,extarrows, mathrsfs, amsthm,amsrefs}
\usepackage[svgnames]{xcolor}
\usepackage{mathdots}
\usepackage{graphicx}
\usepackage{MnSymbol}
\usepackage{tikz}
\usetikzlibrary{calc, positioning}
\usepackage{caption}
\usepackage{adjustbox}

\usepackage[margin=1in, centering]{geometry}
\usepackage[bookmarksnumbered, colorlinks, plainpages]{hyperref}
\usepackage{hyperref} 
\hypersetup{
    colorlinks=true,       
    linkcolor=blue,          
    citecolor=magenta,        
    filecolor=magenta,      
    urlcolor=cyan           
}

\usepackage{mathtools}
\mathtoolsset{showonlyrefs,showmanualtags}

\newtheorem{theorem}{Theorem}[section]

\newtheorem{lemma}[theorem]{Lemma}
\newtheorem{proposition}[theorem]{Proposition}
\newtheorem{remark}[theorem]{Remark}

\newcommand{\bm}{{\bf m}}

\newcommand{\Z}{\mathbb Z}
\newcommand{\R}{\mathbb R}
\newcommand{\C}{\mathbb C}
\newcommand{\T}{\mathbb T}
\newcommand{\N}{\mathbb N}

\definecolor{deepgreen}{cmyk}{1,0,1,0.5}

\title[Magnetic Laplacian on the Lieb lattice]{On the spectrum of magnetic Laplacian on the Lieb lattice}

\author[M.\ Solis]{Moises Gomez Solis}
\address{Department of Mathematics \\ Louisiana State University  \\  Baton Rouge, LA 70803, USA}
\email{mgome29@lsu.edu}

\author[D.\ Spedale]{Dylan Spedale}
\address{Department of Mathematics \\ Louisiana State University  \\  Baton Rouge, LA 70803, USA}
\email{dspeda1@lsu.edu}

\author[F.\ Yang]{Fan Yang}
\address{Department of Mathematics \\ Louisiana State University  \\  Baton Rouge, LA 70803, USA}
\email{yangf@lsu.edu}

\thanks{M.\ Solis was partially supported by NSF DMS-2143369.}

\begin{document}
\begin{abstract}
We study the magnetic Laplacian on the Lieb lattice, and prove Cantor spectrum for arbitrary irrational magnetic flux.
We also provide a complete spectral analysis for the reduced one-dimensional Hamiltonian, proving Cantor spectra for all irrational frequencies, and sharp arithmetic phase transitions. Part of our analysis reveals a novel coexistence phenomenon of point spectrum and absolutely/singular continuous spectrum.
\end{abstract}

\thanks{}

\maketitle

\section{Introduction}

The Lieb lattice is lattice named in honor of Elliot H. Lieb, first appearing in his study of the Hubbard model \cite{Li}. 
This lattice is characterized by a unit cell containing three atomic sites, as shown in Fig.\ref{LiebL}.
In particular, since the $A$-type sites are only connected to the $B,C$-types and vice-versa, it is a bipartite lattice.
This lattice can also be viewed as a variant of the square lattice, by adding the edge-centered $B,C$ vertices in the square lattice consisting solely of $A$ vertices. The $\mathrm{CuO_2}$ plane \cite{Em,SSSW} in copper-oxide superconductors is a well known material that has a natural Lieb lattice structure.

Due to the specific geometry, the study of the Lieb lattice continues to be active with ongoing investigations into its unique and interesting properties, such as flat band \cite{BP, NOA, Mu, Sl}, Dirac cone \cite{NOA, LHID, Sl}, topological insulator behavior \cite{DMM,MMD, WF}, Hofstadter type spectrum \cite{AAM}, quantum spin Hall effect \cite{GUB, BES}, localization \cite{LMZR,Vi}, ferromagnetism \cite{MT,Ta}, and high temperature superconductivity \cite{MMD,JPVKT}. There are also  experimental realizations of the Lieb lattice in various systems such as photonic lattices \cite{Mu, Gu, Vi}, cold atom systems \cite{Taie}, and electronic systems \cite{Sl}, etc. Such research contributes significantly to understanding the quantum mechanical nature of materials and potential applications in technological innovations.

The band structure of the Lieb lattice is of particular interest. While other lattices such as the checkerboard lattice and the Kagome lattice also display flat band, the Lieb lattice is special since its flat band is robust against rational magnetic fields \cite{AAM, NOA}.
More generally, the $2^n$ root model \cite{MMD}, extended Lieb lattice \cite{BP, MZR}, and Lieb lattice with more general hopping integrals or next nearest neighborhood interaction \cite{PM, PSBKM,JPVKT, WF} have been studied extensively.
In this paper, we prove Cantor spectrum for the Lieb lattice model in an arbitrary irrational magnetic field, thus theoretically confirming the fractal structure of the Hofstadter butterfly that was numerically plotted in \cite{AAM}. This result lays the theoretical foundation for the study of quantum Hall effect of the Lieb lattice.
We also provide a complete spectral analysis for the reduced one-dimensional Hamiltonian, proving Cantor spectra for all irrational frequencies, and sharp arithmetic phase transitions. 

We now proceed to describe the tight-binding Hamiltonian that we consider in this paper. 

\subsection{The model}\label{sec:model}
\

\begin{figure} 
\begin{tikzpicture}

    \tikzset{Aset/.style={black,circle,draw,fill=gray,scale=1.6,inner sep=2pt}};
    \tikzset{Bset/.style={black,circle,draw,fill=blue,scale=1.4,inner sep=2pt}};
    \tikzset{Cset/.style={black,circle,draw,fill=red,scale=1.4,inner sep=2pt}};

    \coordinate (origin_of_array) at (0,0);
    \coordinate (bottom_left) at (0,0);
    \coordinate (top_right) at (6,6);

    \node[Aset, label=left:$A$] (a1) at (1,1) {};
    \node[Aset, label=left:$A$] (a2) at (1,3) {};
    \node[Aset, label=left:$A$] (a3) at (1,5) {};
    \node[Aset, label=below:$A$] (a4) at (3,1) {};
    \node[Aset, ] (a5) at (3,3) {};
    \node[Aset, ] (a6) at (3,5) {};
    \node[Aset, label=below:$A$] (a7) at (5,1) {};
    \node[Aset, ] (a8) at (5,3) {};
    \node[Aset, ] (a9) at (5,5) {};
    \node[Aset, label=below:$A$] (a10) at (5,1) {};
    \node[Cset, label=below:$C$] (c1) at (2,1) {};
    \node[Cset, ] (c2) at (2,3) {};
    \node[Cset, ] (c3) at (2,5) {};
    \node[Cset, label=below:$C$] (c5) at (4,1) {};
    \node[Cset, ] (c6) at (4,3) {};
    \node[Cset, ] (c7) at (4,5) {};
    \node[Bset, label=left:$B$] (b1) at (1,2) {};
    \node[Bset, label=left:$B$] (b2) at (1,4) {};
    \node[Bset, ] (b4) at (3,2) {};
    \node[Bset, ] (b5) at (3,4) {};
    \node[Bset, ] (b7) at (5,2) {};
    \node[Bset, ] (b8) at (5,4) {};
    \node[label=left:$na$] at (0,1) {};
    \node[label=left:$(n+1)a$] at (0,3) {};
    \node[label=left:$(n+2)a$] at (0,5) {};
    \node[label=below:$ma$] at (1,-0.13) {};
    \node[label=below:$(m+1)a$] at (3,0) {};
    \node[label=below:$(m+2)a$] at (5,0) {};
    \draw[thick] (a1)--(a3);
    \draw[thick] (a3)--(a9);
    \draw[thick] (a9)--(a7);
    \draw[thick] (a7)--(a1);
    \draw[thick] (a2)--(a8);
    \draw[thick] (a4)--(a6);
    \draw[thick, ->] (0,0)--(0,6){};
    \draw[thick, ->] (0,0)--(6,0){};
 
    \path[<->] (a1) edge[bend right = 60] node[xshift = 4.5pt]{${\scriptscriptstyle t_y}$} (b1); 
    \path[<->] (a1) edge[bend right = 60] node[yshift = -3pt] {${\scriptscriptstyle t_x}$} (c1);
    
\end{tikzpicture}
\caption{Lieb Lattice}
\label{LiebL}
\end{figure}
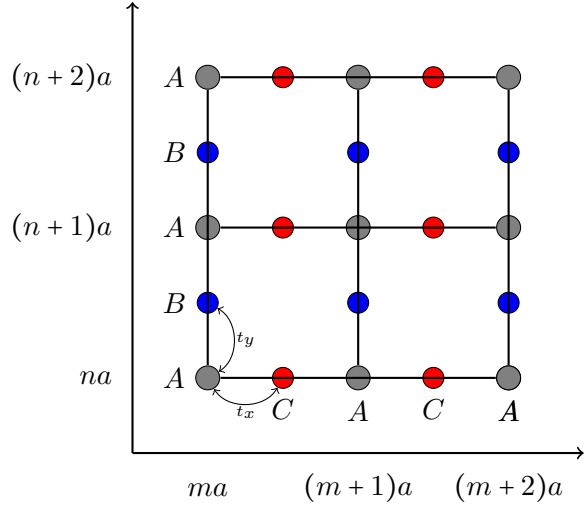
We denote by $t_x$ and $t_y$ the \textit{hopping integrals} along the $x$, and $y$ directions, respectively, see Figure \ref{LiebL}. 
We assume without loss of generality that the hopping integrals satisfy $t_x=t>0$ and $t_y=1$. Then the tight-binding Hamiltonian of the Lieb lattice immersed in a perpendicularly oriented magnetic field $\bm{A} = (-2\pi\alpha\cdot y, 0,0)$ (where $2\pi \alpha$ is the magnetic flux through each unit cell), under the assumption that interaction occurs only between nearest-neighbors, is given by:
\begin{align}\label{LiebHam}
    (H_{\alpha,t}u)^A_{m,n} &=e^{i\pi m \alpha}u^B_{{m,n}} + e^{-i\pi m \alpha} u^B_{{m,n-1}} + t u^C_{{m,n}}+t u^C_{{m-1,n}}\\
    (H_{\alpha,t}u)^B_{{m,n}}&=e^{i\pi m \alpha}u^A_{{m,n+1}}+ e^{-i\pi m \alpha}u^A_{{m,n}}\\
    (H_{\alpha,t}u)^C_{{m,n}}&=t u^A_{{m,n}}+t u^A_{{m+1,n}}
\end{align}
Taking the Fourier transform in the  variable $n$, the two-dimensional Hamiltonian $H_{\alpha,t}$ is reduced to a family of one-dimensional operators $H_{\alpha,t,\theta}$, $\theta\in \T$, acting on 
$$\ell^2(\Z, \C^3)=\{u=(...,u_{m+1}^C,u_m^A,u_m^B,u_m^C,u_{m-1}^A,...)^T, \sum_{m\in \Z}\sum_{D=A,B,C}|u_m^D|^2<\infty\}$$ as follows,
\begin{align}\label{def:H_alpha_theta}
    (H_{\alpha,t,\theta}u)^A_{m}&=e^{\pi i m\alpha}(1+e^{-2\pi i(\theta+m\alpha)})u^B_{m}+t(u^C_m+u^C_{m-1}),\\
    (H_{\alpha,t,\theta}u)^B_m&=e^{-\pi i m\alpha}(1+e^{2\pi i(\theta+m\alpha)})u^A_m\\
    (H_{\alpha,t,\theta}u)^C_m&=t(u^A_m+u^A_{m+1}).
\end{align}
It is well-known that 
\begin{align}\label{eq:sigH^2=sigH^1}
    \sigma(H_{\alpha,t})=\bigcup_{\theta} \sigma(H_{\alpha,t,\theta}).
\end{align}
The operator $H_{\alpha,t,\theta}$ can be written as a $3\times 3$ block operator form:
\begin{align}\label{def:H_att}
    (H_{\alpha,t,\theta}U)_m=BU_{m+1}+V_m(\theta+m\alpha)U_m+B^TU_{m-1},
\end{align}
in which $U_m=(u_m^A, u_m^B, u_m^C)^T$, and
\begin{align}
    B:=\left(\begin{matrix} 0 & 0 &0\\ 0 & 0 &0 \\ t&0 &0\end{matrix}\right), \text{ and } V_m(\theta)=\left(\begin{matrix} 0 &\mathcal{K}_m(\theta) & t\\ \overline{\mathcal{K}_m(\theta)} & 0 &0\\
    t & 0 &0\end{matrix}\right),
\end{align}
where we have defined 
\begin{align}\label{def:Km}
\mathcal{K}_m(\theta):=e^{\pi i m\alpha}(1 + e^{-2\pi i \theta}).
\end{align}
Recently, block-valued operators have received increasing attention. These operators arise naturally from lattices consisting of more than one vertex in each fundamental domain (e.g. multi-layer graphene models).
In general, the study of block-valued operators is much more complicated than the scalar ones since the transfer matrices associated to block-valued operators are symplectic matrices of larger sizes than $\mathrm{SL}(2,\R)$ matrices in the scalar case.
For example, non-perturbative localization in the positive Lyapunov exponent regime was proved for the scalar case more than 20 years ago by Jitomirskaya for the almost Mathieu operator \cite{Ji} and by Bourgain-Goldstein for general analytic potentials \cite{BG}. However, the block-valued case was only established recently by Han and Schlag in \cite{HS1,HS2} in the non-singular setting (when $\det B\neq 0$), resolving a long standing problem since Bourgain and Jitomirskaya in \cite{BJ} and Klein in \cite{Kl}.
Han and Schlag's result does not apply to our model, since $H_{\alpha,t,\theta}$ is singular (meaning $\det B= 0$).
However, singularity works in favor of our spectral analysis. In fact, we are able to utilize this singularity to reduce the $3\times 3$ block-valued $H_{\alpha,t,\theta}$ into an extensively studied scalar model: the almost Mathieu operator (AMO).
This reduction enables us to perform a complete spectral analysis of $H_{\alpha,t,\theta}$.

\subsection{Main results}
\

In the literature, the band structure of the tight binding model on the Lieb lattice has been studied extensively, showing that $E=0$ is one degenerate flat band between dispersive bands for any rational magnetic flux $\alpha$.

Our first result shows that $E=0$ is always in the spectrum for any irrational $\alpha$.
\begin{theorem}\label{thm:zero}
    For any irrational $\alpha$ and any $\theta\in \T$, $0\in \sigma(H_{\alpha,t,\theta})$.
\end{theorem}
 One may be able to prove Theorem \ref{thm:zero} through rational approximations following the argument of \cite[Theorem 3.6]{AS}. We instead give a direct proof by the Weyl's criterion and explicit constructions in Sec.~\ref{sec:zero_proof}.

For any non-integer $\alpha$, the zero energy is gapped away from the {\it bulk of the spectrum}, which we define as:
\begin{align}\label{def:bulk}
\Sigma_{\alpha,t}:=(\bigcup_{\theta\in \T}\sigma(H_{\alpha,t,\theta}))\setminus \{0\}=\sigma(H_{\alpha,t})\setminus \{0\}.
\end{align}
\begin{theorem}\label{thm:away_from_zero}
For $t\neq 0$ and $\alpha\notin \Z$, the bulk of the spectrum $\Sigma_{\alpha,t}$ satisfies $\mathrm{dist}(\Sigma_{\alpha,t}, \{0\})>0$.
Furthermore, for any irrational $\alpha$,
\begin{itemize}
    \item $\Sigma_{\alpha,t}$ is a Cantor set.
    \item for $t=1$, the Hausdorff dimension of $\Sigma_{\alpha,t}$ (and also $\sigma(H_{\alpha,t,\theta})=\Sigma_{\alpha,t}\cup \{0\}$) is at most $1/2$. 
\end{itemize}
\end{theorem}
We present the proof of Theorem \ref{thm:away_from_zero} in Sec. \ref{sec:reduction}.
As we mentioned earlier, the proof relies on a reduction from $H_{\alpha,t,\theta}$ to AMO, similar to that in \cite{BHJ}, which reduces a graphene model into a family of singular Jacobi matrix.
Such reduction uses the ``squaring the Hamiltonian'' technique, which relies crucially on the bipartite nature of the graph.

We also present an alternate reduction via the Weyl criterion, which essentially uses the singularity of the operator ($\det B=0$), in Sec. \ref{sec:Weyl}. This alternate approach is more robust and can be applied to various other lattices beyond bipartite ones, which we will address in follow up works.

Taking \eqref{eq:sigH^2=sigH^1} into account, the two theorems above imply the following concerning the two-dimensional Hamiltonian $H_{\alpha,t}$:
\begin{theorem}\label{thm:Cantor_2d}
For $t\neq 0$ and $\alpha\notin \Z$, $E=0$ is a point spectrum of $H_{\alpha,t}$, and $\mathrm{dist}(\Sigma_{\alpha,t},\{0\})>0$. $H_{\alpha,t}$ has purely continuous spectrum in $\Sigma_{\alpha,t}$.
Furthermore, if $\alpha$ is irrational,
    \begin{itemize}
        \item $\Sigma_{\alpha,t}$ is a Cantor set.
        \item for $t=1$, $H_{\alpha,t}$ has purely singular continuous spectrum in $\Sigma_{\alpha,t}$, and the Hausdorff dimension of $\Sigma_{\alpha,t}$ (and also $\sigma(H_{\alpha,t})=\Sigma_{\alpha,t}\cup \{0\}$) is at most $1/2$.
    \end{itemize}  
\end{theorem}
The point spectrum claim for $E=0$ above follows from the fact that an isolated point in the spectrum cannot support continuous spectral measures.

Our next theorem address the spectral decomposition of the reduced one-dimensional operator $H_{\alpha,t,\theta}$.
The most interesting phenomenon is the coexistence of point spectrum at $E=0$ and singular/absolutely continuous spectrum in the bulk spectrum $\Sigma_{\alpha,t}$ for certain parameters specified below.

To introduce the arithmetic spectral transitions, we define two indices $\beta(\alpha), \gamma(\alpha, \theta)\in [0, +\infty]$ by
\begin{align}
\beta(\alpha)&:=\limsup_{n\rightarrow\infty}\left(-\frac{\ln{\|n\alpha\|_{\T}}}{|n|}\right), \label{defbeta}\\
\gamma(\alpha, \theta)&:=\limsup_{n\rightarrow\infty}\left(-\frac{\ln{\|2\theta+n\alpha\|_{\T}}}{|n|}\right),\label{defgamma}
\end{align}
where $\|x\|_{\T}:=\text{dist} (x, \Z)$.
It is well-known that $\beta(\alpha)=0$ holds for a.e. $\alpha$, and for any fixed $\alpha$, $\gamma(\alpha,\theta)=0$ holds for a.e. $\theta$.

\begin{theorem}\label{thm:spectral_dec}
For $t\neq 0$ and any $\alpha\notin \Z$, any $\theta$, $E=0$ is a point spectrum of $H_{\alpha,t,\theta}$.
For irrational $\alpha$'s, we have the following spectral decompositions.
\begin{itemize}
    \item For $0<t<1$,
\begin{enumerate}
    \item for $\alpha$ such that $0\leq \beta(\alpha)<\log t^{-2}$ and $\theta$ such that $\gamma(\alpha,\theta)=0$, $H_{\alpha,t,\theta}$ exhibits Anderson localization in $\Sigma_{\alpha,t}$.    
    \item for $\alpha$ such that $\beta(\alpha)>\log t^{-2}$ and all $\theta\in \T$, $H_{\alpha,t,\theta}$ has purely singular continuous spectrum in $\Sigma_{\alpha,t}$.
    \item for $\alpha$ such that $\beta(\alpha)=0$ and $\theta$ such that $\gamma(\alpha,\theta)<\log t^{-2}$, $H_{\alpha,t,\theta}$ exhibits Anderson localization in $\Sigma_{\alpha,t}$.
    \item for $\alpha$ such that $\beta(\alpha)=0$ and $\theta$ scuh that $\gamma(\alpha,\theta)>\log t^{-2}$, $H_{\alpha,t,\theta}$ has purely singular continuous spectrum in $\Sigma_{\alpha,t}$.
\end{enumerate}
\item For $t>1$, any irrational $\alpha$ and any $\theta\in\T$, $H_{\alpha,t,\theta}$ has purely absolutely continuous spectrum in $\Sigma_{\alpha,t}$.
\item For $t=1$, any irrational $\alpha$ and any $\theta\in \T$, $H_{\alpha,t,\theta}$ has purely singular continuous spectrum in $\Sigma_{\alpha,t}$.
\end{itemize}
\end{theorem}
One can also show generic point/singular continuous spectrum on the transition lines $\beta(\alpha)=\log t^{-2}$ and $\gamma(\alpha,\theta)=\log t^{-2}$, similar to those results on the almost Mathieu operator \cite{AJZ,JY}, but we don't pursue them here.

The rest of the paper is organized as follows: 
Sec. \ref{sec:AMO} is devoted to the review some relevant results on the almost Mathieu operator. 
We present the proofs of our main theorems in Sec. \ref{sec:proof}, and the results for the Lieb lattice with more general hoppings in Sec. \ref{sec:general}.

\section{Historical results on the almost Mathieu operator}\label{sec:AMO}
In this section, we recall some results on the almost Mathieu opertoat that play important roles in our study of the Lieb lattice.
\subsection{Cantor spectrum}
\

The \textit{Ten Martini problem}, named after Kac and Simon, for the almost Mathieu operator $H^{\mathrm{AMO}}_{\alpha,\theta,t^{-2}}$ conjectures that its spectrum $\sigma(H^{\mathrm{AMO}}_{\alpha,\theta,t^{-2}})$ is a Cantor set for $t\neq 0$ and any irrational $\alpha$.
This problem has been fully resolved, the complete proof was given in Avila-Jitomirskaya \cite{AJ} (for $t\neq 1$) and Avila-Krikorian \cite{AK} (for $t=1$), with earlier breakthroughs by Last \cite{La} and Puig \cite{Pu}.
\begin{theorem}\label{thm:AMO_Cantor}
For $t\neq 0$ and any irrational $\alpha$, $\sigma(H^{\mathrm{AMO}}_{\alpha,\theta,t^{-2}})$ is a Cantor set.
\end{theorem}
In fact, for the $t=1$ case, Last \cite{La} and Avila-Krikorian \cite{AK} proved the following stronger result, which implies Cantor spectrum as a corollary.
\begin{theorem}\label{thm:critical_AMO_mes=1}
    For $t=1$ and any irrational $\alpha$, the Lebesgue measure $\mathrm{mes}(\sigma(H^{\mathrm{AMO}}_{\alpha,\theta,t^{-2}}))=0$.
\end{theorem}
More recently, Jitomirskaya and Krasovsky proved the following upper bound of the Hausdorff dimension of the critical AMO ($t=1$) spectrum.
\begin{theorem}\cite{JK}\label{thm:AMO_Hausdorff}
    For $t=1$ and any irrational $\alpha$, the Hausdorff dimension of $\sigma(H_{\alpha,\theta,1}^{\mathrm{AMO}})$ is at most $1/2$.
\end{theorem}

\subsection{Spectral decompositions}
\

We start by recalling the spectral decomposition of the almost Mathieu operator.
The first ground-breaking result is due to Jitomirskaya \cite{Ji}.
\begin{theorem}\label{thm:AMO_loc}\cite{Ji}
    For $0<t<1$, $\alpha\in \mathrm{DC}$ \footnote{$\mathrm{DC}:=\bigcup_{c>0, a>1}\mathrm{DC}_{c,a}$, where $\mathrm{DC}_{c,a}:=\{\alpha\in \T:\, \|n\alpha\|_{\T}\geq \frac{c}{|n|^a}, \text{ for any } n\neq 0\}$.} and (arithmetically defined) a.e. $\theta$, any non-trivial generalized eigenfunction of $H^{\mathrm{AMO}}_{\alpha,\theta,t^{-2}}v=Ev$ decays exponentially, in particular $H^{\mathrm{AMO}}_{\alpha,\theta,t^{-2}}$ exhibits Anderson localization.
\end{theorem}
The recent years have seen a lot of advances in further understanding the arithmetic spectral transition beyond Diophantine $\alpha\in \mathrm{DC}$. 
The most recent results for the almost Mathieu operators in the positive Lyapunov exponent regime are due to Jitomirskaya and Liu:
\begin{theorem}\cite{JL1}
    For $0<t<1$, 
    \begin{itemize}
        \item for irrational $\alpha$ such that $0\leq \beta(\alpha)<\log t^{-2}$, and $\theta$ such that $\gamma(\alpha,\theta)=0$, any non-trivial generalized eigenfunction of $H^{\mathrm{AMO}}_{\alpha,\theta,t^{-2}}v=Ev$ decays exponentially, in particular $H^{\mathrm{AMO}}_{\alpha,\theta, t^{-2}}$ exhibits Anderson localization. 
        \item for irrational $\alpha$ such that $\beta(\alpha)>\log t^{-2}$ and all $\theta\in \T$, $H^{\mathrm{AMO}}_{\alpha,\theta, t^{-2}}$ has purely singular continuous spectrum.
    \end{itemize}
\end{theorem}
\begin{remark}
    The singular continuous part and localization part (with non-arithmetic condition for $\theta$) was first proved by Avila-You-Zhou in \cite{AYZ}.
\end{remark}

\begin{theorem}\cite{JL2}\label{thm:AMO_loc2}
    For $0<t<1$,
    \begin{itemize}
        \item for irrational $\alpha$ such that $\beta(\alpha)=0$ and $\theta$ such that $\gamma(\alpha,\theta)<\log t^{-2}$, any non-trivial generalized eigenfunction of $H^{\mathrm{AMO}}_{\alpha,\theta,t^{-2}}v=Ev$ decays exponentially, in particular $H^{\mathrm{AMO}}_{\alpha,\theta,t^{-2}}$ exhibits Anderson localization.
        \item for irrational $\alpha$ such that $\beta(\alpha)=0$ and $\theta$ such that $\gamma(\alpha,\theta)>\log t^{-2}$, 
        $H^{\mathrm{AMO}}_{\alpha,\theta,t^{-2}}$ has purely singular continuous spectrum.
    \end{itemize}
\end{theorem}

Purely absolutely continuous spectrum was proved by Avila throughout the entire $t>1$ regime.
\begin{theorem}\cite{A_ac}\label{thm:AMO_ac}
    For $t>1$, and any irrational $\alpha$ and any $\theta\in \T$, the spectral measures of $H^{\mathrm{AMO}}_{\alpha,\theta,t^{-2}}$ are absolutely continuous.
\end{theorem}

For $t=1$, Avila-Jitomirskaya-Marx (for non $\alpha$-rational $\theta$, those $2\theta\notin \Z\alpha+\Z$) and Jitomirskaya (for $\alpha$-rational $\theta$, those $2\theta\in \Z\alpha+\Z$) proved purely singular continuous spectrum.
\begin{theorem}\cite{AJM,J_sc}\label{thm:AMO_critical}
    For $t=1$, any irrational $\alpha$ and any $\theta\in \T$, $H^{\mathrm{AMO}}_{\alpha,\theta,t^{-2}}$ has purely singular continuous spectrum.
\end{theorem}

\section{Proofs of theorems}\label{sec:proof}
\subsection{The zero energy: Theorem \ref{thm:zero}}\label{sec:zero_proof}
\begin{proof}
    Note that $1+e^{2\pi i\theta}=0$ if $\theta=1/2$. 
    For any given $\theta$ and irrational $\alpha$, for any $k\in \N\setminus \{0\}$, by the denseness of the orbit $\{\theta+m\alpha\}_{m\in \Z}$, there exists $m_k\in \N$ such that 
    \begin{align}
        |\mathcal{K}_{m_k}(\theta+m_k\alpha)|=|1+e^{2\pi i(\theta+m_k\alpha)}|<\frac{1}{k}.
    \end{align}
    Now we construct an approximate eigenfunction. Let $u$ be defined as
    \begin{align}
    \begin{cases}
        u_{m}^{A,B,C}=0, \text{ if } m\neq m_k\\
        u_{m_k}^A=u_{m_k}^C=0\\
        u_{m_k}^B=1
    \end{cases}
    \end{align}
    Clearly such $u$ satisfies $\|u\|_{\ell^2(\Z,\C^3)}=1$. One can check by explicit computations that
    \begin{align}
    \begin{cases}
        (H_{\alpha,t,\theta}u)_{m_k}^A=\mathcal{K}_{m_k}(\theta+m_k\alpha), \\
        (H_{\alpha,t,\theta}u)_{m_k}^{B,C}=0\\
        (H_{\alpha,t,\theta}u)_m^{A,B,C}=0, \text{ if } m\neq m_k
    \end{cases}
    \end{align}
    This implies
    \begin{align}
        \|H_{\alpha,t,\theta}u\|_{\ell^2(\Z,\C^3)}<\frac{1}{k}, \text{ while } \|u\|_{\ell^2(\Z,\C^3)}=1.
    \end{align}
    Therefore the Weyl's criterion yields $0\in \sigma(H_{\alpha,t,\theta})$ for any $\theta$.
\end{proof}

\subsection{Reduction to the almost Mathieu operator}\label{sec:reduction}
\

To see the connection to AMO, one can view $H_{\alpha,t,\theta}$ as follows:
\begin{align}
    H_{\alpha,t,\theta}=\left(\begin{matrix}
        0 &\tilde{H}_{\alpha,t,\theta}\\
        \tilde{H}_{\alpha,t,\theta}^* &0
    \end{matrix}\right),
\end{align}
in which $\tilde{H}_{\alpha,t,\theta}$: $\ell^2(\Z,\C^2)\to \ell^2(\Z,\C)$ is defined as 
\begin{align}
    (\tilde{H}_{\alpha,t,\theta}u)_m=e^{\pi im\alpha}(1+e^{-2\pi i(\theta+m\alpha)})u_m^B+t(u_m^C+u_{m-1}^C),
\end{align}
and $\tilde{H}_{\alpha,t,\theta}^*: \ell^2(\Z,\C)\to \ell^2(\Z,\C^2)$ is:
\begin{align}
    (\tilde{H}_{\alpha,t,\theta}^*u)_m^B=&e^{-\pi im\alpha}(1+e^{2\pi i(\theta+m\alpha)})u_m\\
    (\tilde{H}_{\alpha,t,\theta}^*u)_m^C=&t(u_m+u_{m+1}).
\end{align}
By squaring the Hamiltonian $H_{\alpha,t,\theta}$, one arrives at
\begin{align}\label{eq:H2=diagonal}
    H_{\alpha,t,\theta}^2=
    \left(\begin{matrix}
        \tilde{H}_{\alpha,t,\theta}\tilde{H}^*_{\alpha,t,\theta} & 0\\
        0 & \tilde{H}_{\alpha,t,\theta}^*\tilde{H}_{\alpha,t,\theta}
    \end{matrix}\right).
\end{align}
A simple computation shows that $\tilde{H}_{\alpha,t,\theta}\tilde{H}^*_{\alpha,t,\theta}$ acting on $\ell^2(\Z)$ is:
\begin{align}\label{eq:HH*=HAMO}
    (\tilde{H}_{\alpha,t,\theta}\tilde{H}^*_{\alpha,t,\theta}u)_m
    =&\mathcal{K}_m(\theta+m\alpha)(\tilde{H}_{\alpha,t,\theta}^*u)_m^B+t((\tilde{H}_{\alpha,t,\theta}^*u)_m^C+(\tilde{H}_{\alpha,t,\theta}^*u)_{m-1}^C)\notag\\
    =&|\mathcal{K}_m(\theta+m\alpha)|^2 u_m+t^2(u_{m+1}+2u_m+u_{m-1})\notag\\
    =&t^2(u_{m+1}+u_{m-1}+2u_m +2t^{-2}(1+\cos(2\pi(\theta+m\alpha)))u_m)\notag\\
    =&t^2((H^{\mathrm{AMO}}_{\alpha,\theta,t^{-2}}u)_m+(2+2t^{-2})u_m),
\end{align}
in which $\mathcal{K}_m$ is as in \eqref{def:Km}.

It is well-known that
\begin{align}\label{eq:uni_equiv}
    (\tilde{H}_{\alpha,t,\theta}\tilde{H}_{\alpha,t,\theta}^*)|_{(\mathrm{ker}(\tilde{H}_{\alpha,t,\theta}^*))^\perp} \text{ and } (\tilde{H}_{\alpha,t,\theta}^*\tilde{H}_{\alpha,t,\theta})|_{(\mathrm{ker}(\tilde{H}_{\alpha,t,\theta}))^\perp} \text{ are unitarily equivalent}.
\end{align}
One can show that 
\begin{align}\label{eq:ker=empty}
\mathrm{ker}(\tilde{H}_{\alpha,t,\theta}^*)=\emptyset.
\end{align}
In fact suppose $\tilde{H}_{\alpha,t,\theta}^*u=0$, $\|u\|_{\ell^2(\Z)}=1$, then 
\begin{align}
    t(u_m+u_{m+1})=0, \text{ for any } m\in \Z.
\end{align}
This implies $u=0$ since otherwise $\|u\|_{\ell^2(\Z)}=\infty$.

Combining \eqref{eq:uni_equiv} with \eqref{eq:ker=empty}, we have
\begin{align}\label{eq:uni_equiv_no_ker}
    \tilde{H}_{\alpha,t,\theta}\tilde{H}_{\alpha,t,\theta}^* \text{ and } (\tilde{H}_{\alpha,t,\theta}^*\tilde{H}_{\alpha,t,\theta})|_{(\mathrm{ker}(\tilde{H}_{\alpha,t,\theta}))^\perp} \text{ are unitarily equivalent}.
\end{align}
This implies
\begin{align}\label{eq:sig_HH*=H*H_no_zero}
    \sigma(\tilde{H}_{\alpha,t,\theta}\tilde{H}_{\alpha,t,\theta}^*)                                                                                                                                                                                                                                                                                                                =\sigma(\tilde{H}_{\alpha,t,\theta}^*\tilde{H}_{\alpha,t,\theta})\setminus \{0\}.
\end{align}
It is clear from \eqref{eq:HH*=HAMO} that
\begin{lemma}\label{lem:Lieb_to_AMO}
    \begin{align}
\sigma(\tilde{H}_{\alpha,t,\theta}\tilde{H}_{\alpha,t,\theta}^*)=t^2\cdot \sigma(H^{\mathrm{AMO}}_{\alpha,\theta,t^{-2}})+2+2t^2,
    \end{align}
in which $a\cdot \mathcal{B}+b:=\{ax+b, x\in \mathcal{B}\}$ for a set $\mathcal{B}$.
\end{lemma}

The following is well-known: 
\begin{theorem}\label{thm:AMO_gap_boundary}
There exists some positive constant $c_{\alpha,t}>0$ such that for any $\alpha\notin \Z$ and any $\theta\in \T$:
    \begin{align}
    \sigma(H^{\mathrm{AMO}}_{\alpha,\theta,t^{-2}})\subseteq \left[-2-2t^{-2}+c_{\alpha,t}, 2+2t^{-2}-c_{\alpha,t}\right].
\end{align}
\end{theorem}
We include a short proof below for completeness.
\begin{proof}
    It is enough to show that the operator norm 
\begin{align}\label{eq:norm_H_AMO^2}
    \|(H^{\mathrm{AMO}}_{\alpha,\theta,t^{-2}})^2\|\leq (2+2t^{-2})^2-c_{\alpha,t}',
\end{align}
for some positive constant $c_{\alpha,t}'$.
A simple computation and using triangle inequality yields
\begin{align}
    \|(H^{\mathrm{AMO}}_{\alpha,\theta,t^{-2}})^2\|
    \leq &4+4t^{-4}+4t^{-2}\sup_{m\in \Z}(\cos(2\pi(\theta+m\alpha))+\cos(2\pi(\theta+(m+1)\alpha))\\
    \leq &4+4t^{-2}+4t^{-2}(2-c_{\alpha}),
\end{align}
for some positive constant $c_{\alpha}>0$ due to $\alpha\notin \Z$. This clearly yields \eqref{eq:norm_H_AMO^2}, hence the claimed result.
\end{proof}

Theorem \ref{thm:AMO_gap_boundary} together with Lemma \ref{lem:Lieb_to_AMO} implies $\bigcup_{\theta\in \T}\sigma(\tilde{H}_{\alpha,t,\theta}\tilde{H}_{\alpha,t,\theta}^*)$ is {\it bounded away} from $0$, which is crucial for the rest of the paper.
\begin{proposition}\label{prop:away_from_zero}
There exists constants $C_{1,\alpha,t}>C_{2,\alpha,t}>0$ such that
\begin{align}
    \bigcup_{\theta\in \T}\sigma(\tilde{H}_{\alpha,t,\theta}\tilde{H}_{\alpha,t,\theta}^*)\subset [C_{2,\alpha,t}, C_{1,\alpha,t}]\subset (0,\infty).
\end{align}
\end{proposition}

Combining Proposition \ref{prop:away_from_zero} with \eqref{eq:sig_HH*=H*H_no_zero}, we have
\begin{align}
    \sigma(\tilde{H}_{\alpha,t,\theta}^*\tilde{H}_{\alpha,t,\theta})\setminus \{0\}=\sigma(\tilde{H}_{\alpha,t,\theta}\tilde{H}_{\alpha,t,\theta}^*)\subset [C_{2,\alpha,t},C_{1,\alpha,t}]\subset (0,\infty).
\end{align}
This implies either $E=0$ is a point spectrum of $\tilde{H}^*_{\alpha,t,\theta}\tilde{H}_{\alpha,t,\theta}$ or not a spectrum.
By the same argument as in the proof of Theorem \ref{thm:zero}, one can show that $0\in\sigma(\tilde{H}^*_{\alpha,t,\theta}\tilde{H}_{\alpha,t,\theta})$.
Hence we have
\begin{lemma}\label{lem:sigH*H=sigHH*_0}
\begin{align}
    \sigma(\tilde{H}_{\alpha,t,\theta}^*\tilde{H}_{\alpha,t,\theta})=\sigma(\tilde{H}_{\alpha,t,\theta}\tilde{H}_{\alpha,t,\theta}^*)\cup \{0\}\subset [C_{2,\alpha,t},C_{1,\alpha,t}]\cup \{0\}.
\end{align}
Furthermore $E=0$ is a point spectrum of $\tilde{H}_{\alpha,t,\theta}^*\tilde{H}_{\alpha,t,\theta}$ for any $t\neq 0$, any $\alpha\notin\Z$ and any $\theta\in \T$.
\end{lemma}
Our next theorem shows the symmetry of $\sigma(H_{\alpha,t,\theta})$.
\begin{theorem}\label{thm:H=-H}
$H_{\alpha,t,\theta}$ is unitarily equivalent to $-H_{\alpha,t,\theta}$. As a consequence, we have
\begin{align}
\sigma(H_{\alpha,t,\theta})=-\sigma(H_{\alpha,t,\theta}).
\end{align}
\end{theorem}
\begin{proof}
    Let $U$ be such that
    \begin{align}
        (Uu)_m^A=-u_m^A, \text{ and } (Um)_m^{B,C}=u_m^{B,C}, \text{ for any } m\in Z.
    \end{align}
    Clearly $U$ is unitary. It is also easy to check that $U^*HU=H$, hence the claimed result.
\end{proof}

Combining \eqref{eq:H2=diagonal} with Lemmas \ref{lem:Lieb_to_AMO}, \ref{lem:sigH*H=sigHH*_0}, Theorem \ref{thm:H=-H}, we have
\begin{theorem}\label{thm:sig_H}
    \begin{align}\label{eq:Sigma=sqrt}
        \sigma(H_{\alpha,t,\theta})
        =&\pm\sqrt{\sigma(\tilde{H}_{\alpha,t,\theta}\tilde{H}_{\alpha,t,\theta}^*)}\cup\{0\}\\
        =&\pm\sqrt{t^2\cdot \sigma(H^{\mathrm{AMO}}_{\alpha,\theta,t^{-2}})+2+2t^2}\cup \{0\}.
    \end{align}
    In particular, the bulk of spectrum $\Sigma_{\alpha,t}$, defined as in \eqref{def:bulk} satisfies $\mathrm{dist}(\Sigma_{\alpha,t},0)>0$. 
\end{theorem}
The later claim follows from combining \eqref{eq:Sigma=sqrt} with Lemma \ref{lem:sigH*H=sigHH*_0}.

\subsection{Proofs of Theorems \ref{thm:away_from_zero} and \ref{thm:Cantor_2d}}\ 

The Cantor set assertions in Theorems \ref{thm:away_from_zero} and \ref{thm:Cantor_2d} follow from combining Theorem \ref{thm:sig_H} with Theorem \ref{thm:AMO_Cantor}.
The claims on the Hausdorff dimensions follow from combining Theorem \ref{thm:sig_H} with \ref{thm:AMO_Hausdorff}.

To show purely continuous spectrum of $H_{\alpha,t}$, we need to show the two dimensional Hamiltonian $H_{\alpha,t}$ can not have an eigenvalue $E\neq 0$.
Assume to the contrary that $H_{\alpha,t}u=Eu$, $E\neq 0$, has a solution $u=(u_{m,n}^{A,B,C})\in \ell^2(\Z^2,\C^3)$ with $\|u\|_{\ell^2(\Z^2,\C^3)}=1$.
Taking Fourier transform as in the derivation of \eqref{eq:sigH^2=sigH^1}, we have for a.e. $\theta\in \T$, 
\begin{align}\label{eq:Hutheta=Eutheta}
    H_{\alpha,t,\theta}u_{\theta}=Eu_{\theta},
\end{align}
where $u_{\theta}=(u_{\theta,m}^{A,B,C})\in \ell^2(\Z,\C^3)$, defined by $u_{\theta,m}^{A,B,C}=\sum_{n\in \Z}u_{m,n}^{A,B,C}e^{2\pi in\theta}$, satisfies $$\int_{\T}\|u_{\theta}\|_{\ell^2(\Z,\C^3)}^2\, \mathrm{d}\theta=1.$$
Note that this implies there is a positive measure set $\mathcal{A}$ of $\theta$ such that $\|u_{\theta}\|_{\ell^2(\Z,\C^3)}>0$.

By \eqref{eq:Hutheta=Eutheta}, for a.e. $\theta\in \T$, we have
\begin{align}
    u_{\theta,m}^B=&E^{-1}\overline{\mathcal{K}_m(\theta+m\alpha)}u_{\theta,m}^A,\\
    u_{\theta,m}^C=&E^{-1}t(u_{\theta,m}^A+u_{\theta,m+1}^A),
\end{align}
and by replacing $u_{\theta}^B, u_{\theta}^C$ with $u_{\theta}^A$, we have
\begin{align}\label{eq:H_AMO_utheta=Eutheta}
    H^{\mathrm{AMO}}_{\alpha,\theta,t^{-2}}u_{\theta}^A=(E^2t^{-2}-2t^{-2}-2)u_{\theta}^A,
\end{align}
with $\|u_{\theta}^A\|_{\ell^2(\Z)}\in (0,1]$ for $\theta\in \mathcal{A}$. 
This implies $H^{\mathrm{AMO}}_{\alpha,\theta,t^{-2}}$ has a {\it fixed} eigenvalue $E^2t^{-2}-2t^{-2}-2$ for $\theta\in \mathcal{A}$, a positive measure set. This contradicts a well-known result, e.g. the continuity of the integrated density of states of the almost Mathieu operator \cite[Theroem 2.9]{AS}.

Finally, to show $H_{\alpha,t=1}$ has purely singular continuous spectrum, it suffices to exclude absolutely continuous spectrum. By Theorem \ref{thm:critical_AMO_mes=1}, $\mathrm{mes}(\sigma(H^{\mathrm{AMO}}_{\alpha,\theta,1}))=0$. This together with Theorem \ref{thm:sig_H} implies $\mathrm{mes}(\Sigma_{\alpha,t})=0$. 
Hence the absolutely continuous spectrum of $H_{\alpha,t=1}$ is empty. \qed

\subsection{Proof of Theorem \ref{thm:spectral_dec}}\ 

For any $t\neq 0$, any $\alpha\notin \Z$ and any $\theta\in \T$, $E=0$ is a point spectrum since it is an isolated point in the spectrum.

To address the spectral decomposition of $H_{\alpha,t,\theta}$ on $\Sigma_{\alpha,t}$, it is enough to study $H^2_{\alpha,t,\theta}$, which through \eqref{eq:H2=diagonal} is reduced to the study of $\tilde{H}_{\alpha,t,\theta}\tilde{H}_{\alpha,t,\theta}^*$ and $\tilde{H}^*_{\alpha,t,\theta}\tilde{H}_{\alpha,t,\theta}$.
By \eqref{eq:HH*=HAMO}, the spectral decomposition of $\tilde{H}_{\alpha,t,\theta}\tilde{H}_{\alpha,t,\theta}^*$ is the same as that of the almost Mathieu operator $H_{\alpha,\theta,t^{-2}}^{\mathrm{AMO}}$.
By \eqref{eq:uni_equiv_no_ker} and Lemma \ref{lem:sigH*H=sigHH*_0}, the spectral decomposition of $\tilde{H}^*_{\alpha,t,\theta}\tilde{H}_{\alpha,t,\theta}$, restricted to $[C_{2,\alpha,t},C_{1,\alpha,t}]\subset (0,\infty)$, is the same as that of $\tilde{H}_{\alpha,t,\theta}\tilde{H}_{\alpha,t,\theta}^*$.
Since by Theorem \ref{thm:AMO_ac}, when $t>1$, $H^{\mathrm{AMO}}_{\alpha,\theta,t^{-2}}$ has purely absolutely continuous spectrum, thus $H_{\alpha,t,\theta}$ also has absolutely continuous spectrum on the bulk spectrum $\Sigma_{\alpha,t}$ for any irrational $\alpha$ and any $\theta\in \T$.
The singular continuous/pure point claims on $\Sigma_{\alpha,t}$ follow similarly when incorporating the singular continuous/pure point results of AMO in Theorems \ref{thm:AMO_loc}, \ref{thm:AMO_loc2} and \ref{thm:AMO_critical}.
Finally, by Shnol's theorem \cite{Be,Simon,Han}, to show $H_{\alpha,t,\theta}$ exhibits Anderson localization in the claimed regimes, it suffices to show any generalized (polynomially bounded) eigenfunction decays exponentially. We assume $u$ is a generalized eigenfunction $H_{\alpha,t,\theta}u=Eu$ for $E\neq 0$.
Direct computations as in \eqref{eq:H_AMO_utheta=Eutheta} show that $u^A$ is a generalized eigenfunction satisfying:
\begin{align}
    H^{\mathrm{AMO}}_{\alpha,\theta,t^{-2}}u^A=(E^2t^{-2}-2t^{-2}-2)u^A.
\end{align}
The exponential decay of $u^A_m$, in $m$, follows from the Anderson localization of $H^{\mathrm{AMO}}_{\alpha,\theta,t^{-2}}$ as in Theorems \ref{thm:AMO_loc} and \ref{thm:AMO_loc2}.
The claimed decay of $u^B$, $u^C$ follows from that of $u^A$ via the following:
    \begin{align}\label{eq:uBC=uA}
        u_m^B=E^{-1} \overline{\mathcal{K}_m(\theta+m\alpha)} u_m^A, \text{ and } u_m^C=E^{-1}t(u_m^A+u_{m+1}^A).
    \end{align}
This proves Anderson localization. \qed

\subsection{An alternate reduction from $H_{\alpha,t,\theta}$ to $H^{\mathrm{AMO}}_{\alpha,\theta,t^{-2}}$}\label{sec:Weyl}
\

First, let $0\neq E\in \sigma(H_{\alpha,t,\theta})$.
Then by the Weyl's criterion, for any small $\varepsilon>0$, there exists a normalized vectors $u$, $\|u\|_{\ell^2(\Z,\C^3)}=1$, such that
\begin{align}\label{eq:L_W1}
    \|(H_{\alpha,t,\theta}-E)u\|_{\ell^2(\Z,\C^3)}\leq \varepsilon.
\end{align}
For $D=A,B,C$, denote 
\begin{align}
    r^{D}_m:=(H_{\alpha,t,\theta}u)^{D}_m-Eu^{D}_m,
\end{align}
and $r=(...,r_m^{A},r_m^{B},r_m^{C},...)^T$.
Thus \eqref{eq:L_W1} implies
\begin{align}\label{eq:L_W2}
    \|r\|_{\ell^2(\Z,\C^3)}\leq \varepsilon.
\end{align}
We have
\begin{align}
Eu^A_m+r^A_m&=\mathcal{K}_m(\theta+m\alpha)u^B_{m}+t(u^C_m+u^C_{m-1}), \label{eq:uA=uB+uC+r}\\
u^B_m&=E^{-1}\overline{\mathcal{K}_m(\theta+m\alpha)} u^A_m-E^{-1}r^B_m \label{eq:uB=uA+r}\\
u^C_m&=E^{-1}t(u^A_m+u^A_{m+1})-E^{-1}r^C_m. \label{eq:uC=uA+r}
\end{align}
Plugging \eqref{eq:uB=uA+r} and \eqref{eq:uC=uA+r} into the \eqref{eq:uA=uB+uC+r}, we have
\begin{align}
    Eu_m^A+r_m^A=E^{-1}|\mathcal{K}_m(\theta+m\alpha)|^2u_m^A-E^{-1}\mathcal{K}_m(\theta+m\alpha)r_m^B+&E^{-1}t^2(u_{m+1}^A+2u_m^A+u_{m-1}^A)\\
    &-E^{-1}t(r_m^C+r_{m-1}^C),
\end{align}
which reduces to 
\begin{align}
    (H^{\mathrm{AMO}}_{\alpha,\theta,t^{-2}}u)_m^A-g_t(E)u_m^A=Et^{-2}r_m^A+t^{-2}\mathcal{K}_m(\theta+m\alpha)r_m^B+t^{-1}(r_m^C+r_{m-1}^C),
\end{align}
in which
$$g_t(E):=E^2t^{-2}-2t^{-2}-2.$$
Hence
\begin{align}\label{eq:L_W3}
    &\|(H^{\mathrm{AMO}}_{\alpha,\theta,t^{-2}}-g_t(E))u^A\|_{\ell^2(\Z)}\\
    \leq &|E|t^{-2} \|r^A\|_{\ell^2(\Z)}+t^{-2} \sup_{\theta}|\mathcal{K}_m(\theta)|\cdot \|r^B\|_{\ell^2(\Z)}+2t^{-2} \|r^C\|_{\ell^2(\Z)}\leq C_{E,t}\, \varepsilon,
\end{align}
where $C_{E,t}>0$ is a constant depending only on $E,t$.
It remains to obtain a lower bound for $\|u^A\|_{\ell^2(\Z)}$.
By \eqref{eq:uB=uA+r} and \eqref{eq:uC=uA+r}, we have
\begin{align}
    \|u^B\|_{\ell^2(\Z)}\leq C_{E,t} \|u^A\|_{\ell^2(\Z)}+E^{-1} \|r^B\|_{\ell^2(\Z)},\\
    \|u^C\|_{\ell^2(\Z)}\leq C_{E,t} \|u^A\|_{\ell^2(\Z)}+E^{-1} \|r^C\|_{\ell^2(\Z)}.
\end{align}
Hence combining the above with that $\|u\|_{\ell^2(\Z,\C^3)}=1$, we have
\begin{align}
    1\leq \|u^A\|_{\ell^2(\Z)}+\|u^B\|_{\ell^2(\Z)}+\|u^C\|_{\ell^2(\Z)}\leq C_{E,t} \|u^A\|_{\ell^2(\Z)}+2E^{-1}\varepsilon.
\end{align}
Hence 
\begin{align}\label{eq:L_W4}
\|u^A\|_{\ell^2(\Z)}\geq c_{E,t}>0.
\end{align}
Combining \eqref{eq:L_W3} with \eqref{eq:L_W4}, we have
\begin{align}
    \frac{\|(H^{\mathrm{AMO}}_{\alpha,\theta,t^{-2}}-g_t(E))u^A\|_{\ell^2(\Z)}}{\|u^A\|_{\ell^2(\Z)}}\leq C_{E,t}\, \varepsilon.
\end{align}
Letting $\varepsilon\to 0$, by the Weyl criterion, $g_t(E)\in \sigma(H^{\mathrm{AMO}}_{\alpha,\theta,t^{-2}})$.

Next, let $g_t(E)\in \sigma(H^{\mathrm{AMO}}_{\alpha,\theta,t^{-2}})$ for some $E\neq 0$. 
By the Weyl's criterion, for any small $\varepsilon>0$, there exists a normalized vector $v$, $\|v\|_{\ell^2(\Z)}=1$, such that 
\begin{align}
    \|(H^{\mathrm{AMO}}_{\alpha,\theta,t^{-2}}-g_t(E))v\|\leq \varepsilon.
\end{align}
Let $\widetilde{r}_m:=(H^{\mathrm{AMO}}_{\alpha,\theta,t^{-2}}v)_m-g_t(E)v_m$, hence $\|r\|_{\ell^2(\Z)}\leq \varepsilon$.
For $m\in \Z$, define
\begin{align}
    u_m^A&=v_m\\
    u_m^B&=E^{1}\overline{\mathcal{K}_m(\theta+m\alpha)}v_m=E^{-1}\overline{\mathcal{K}_m(\theta+m\alpha)}u_m^A\label{eq:uB=uA'}\\
    u_m^C&=E^{-1}t(v_m+v_{m+1})=E^{-1}(u_m^A+u_{m+1}^A).\label{eq:uC=uA'}
\end{align}
Plugging the above into 
\begin{align}
    \widetilde{r}_m&=v_{m+1}+v_{m-1}+2t^{-2}\cos(2\pi (\theta+m\alpha))v_m-(E^2t^{-2}-2t^{-2}-2)v_m,\\
    &=v_{m+1}+v_{m-1}+t^{-2}|\mathcal{K}_m(\theta+m\alpha)|^2 v_m-E^2t^{-2}v_m+2v_m\\
    &=Et^{-1}(u_m^C+u_{m+1}^C)+Et^{-2}\mathcal{K}_m(\theta+m\alpha)u_m^B-E^2t^{-2}u_m^A,
\end{align}
we have
\begin{align}\label{eq:uA=uB+uC'}
    Eu_m^A=\mathcal{K}_m(\theta+m\alpha)u_m^B+t(u_m^C+u_{m+1}^C)-E^{-1}t^2 \widetilde{r}_m.
\end{align}
Combining \eqref{eq:uB=uA'}, \eqref{eq:uC=uA'} with \eqref{eq:uA=uB+uC'}, we have 
\begin{align}
    \|(H_{\alpha,t,\theta}-E)u\|_{\ell^2(\Z,\C^3)}\leq E^{-1}t^2 \|\widetilde{r}\|_{\ell^2(\Z)}\leq E^{-1}t^2 \varepsilon.
\end{align}
Since $\|u\|_{\ell^2(\Z,\C^3)}\geq \|u^A\|_{\ell^2(\Z)}=\|v\|_{\ell^2(\Z)}=1$, we have
\begin{align}
    \frac{\|(H_{\alpha,t,\theta}-E)u\|_{\ell^2(\Z,\C^3)}}{\|u\|_{\ell^2(\Z,\C^3)}}\leq E^{-1}t^2 \varepsilon.
\end{align}
Letting $\varepsilon\to 0$, by the Weyl's criterion, we have $E\in \sigma(H_{\alpha,t,\theta})$.
This completes the proof. \qed

\section{Lieb lattice with general coupling}\label{sec:general}

In this last section, we briefly discuss the Lieb lattice with  more general hopping integrals, see Fig.\ref{DLieb}. 
Let ${\bf t}=(t_1,t_2,t_3, t_4)\in \R_+^4$, without loss of generality we can set $t_1=1$. The 2D tight binding model on the Lieb lattice under magnetic field can be expressed as below, again, only nearest neighborhood interactions are taken into account.
\begin{align}\label{LiebHam_2}
    (H_{\alpha,{\bf t}}^g u)^A_{n,m} &=e^{-i\pi m \alpha}u^B_{{n,m}}+t_4 e^{i\pi m \alpha} u^B_{{n-1,m}}+t_2 u^C_{{n,m}}+t_3 u^C_{{n,m-1}}\\
    (H_{\alpha,{\bf t}}^g u)^B_{{n,m}}&=t_4 e^{-i\pi m \alpha}u^A_{{n+1,m}}+ e^{i\pi m \alpha}u^A_{{n,m}}\\
    (H_{\alpha,{\bf t}}^g u)^C_{{n,m}}&=t_2 u^A_{{n,m}}+t_3 u^A_{{n,m+1}},
\end{align}
where $2\pi \alpha$ denotes the flux through a unit cell of the Lieb lattice. 

This general hopping model and its special cases 
have been studied both theoretically and experimentally in \cite{BGMSJ, WF, PSBKM, JPVKT, BES, PM, RL}.

\begin{figure} 
\begin{tikzpicture}

    \tikzset{Aset/.style={black,circle,draw,fill=gray,scale=1.6,inner sep=2pt}};
    \tikzset{Bset/.style={black,circle,draw,fill=blue,scale=1.4,inner sep=2pt}};
    \tikzset{Cset/.style={black,circle,draw,fill=red,scale=1.4,inner sep=2pt}};

    \coordinate (origin_of_array) at (0,0);
    \coordinate (bottom_left) at (0,0);
    \coordinate (top_right) at (6,6);
    \node[Aset, label=left:$A$] (a1) at (1,1) {};
    \node[Aset, label=left:$A$] (a2) at (1,3) {};
    \node[Aset, label=left:$A$] (a3) at (1,5) {};
    \node[Aset, label=below:$A$] (a4) at (3,1) {};
    \node[Aset, ] (a5) at (3,3) {};
    \node[Aset, ] (a6) at (3,5) {};
    \node[Aset, label=below:$A$] (a7) at (5,1) {};
    \node[Aset, ] (a8) at (5,3) {};
    \node[Aset, ] (a9) at (5,5) {};
    \node[Aset, label=below:$A$] (a10) at (5,1) {};
    \node[Cset, label=below:$C$] (c1) at (2,1) {};
    \node[Cset, ] (c2) at (2,3) {};
    \node[Cset, ] (c3) at (2,5) {};
    \node[Cset, label=below:$C$] (c5) at (4,1) {};
    \node[Cset, ] (c6) at (4,3) {};
    \node[Cset, ] (c7) at (4,5) {};
    \node[Bset, label=left:$B$] (b1) at (1,2) {};
    \node[Bset, label=left:$B$] (b2) at (1,4) {};
    \node[Bset, ] (b4) at (3,2) {};
    \node[Bset, ] (b5) at (3,4) {};
    \node[Bset, ] (b7) at (5,2) {};
    \node[Bset, ] (b8) at (5,4) {};
    \node[label=left:$na$] at (0,1) {};
    \node[label=left:$(n+1)a$] at (0,3) {};
    \node[label=left:$(n+2)a$] at (0,5) {};
    \node[label=below:$ma$] at (1,-0.13) {};
    \node[label=below:$(m+1)a$] at (3,0) {};
    \node[label=below:$(m+2)a$] at (5,0) {};
    \draw[thick] (a1)--(a3);
    \draw[thick] (a3)--(a9);
    \draw[thick] (a9)--(a7);
    \draw[thick] (a7)--(a1);
    \draw[thick] (a2)--(a8);
    \draw[thick] (a4)--(a6);
    \draw[thick,->] (0,0)--(0,6){};
    \draw[thick, ->] (0,0)--(6,0) {};
    
    \path[<->] (a1) edge[bend right = 60] node[xshift = 4.5 pt]{${\scriptscriptstyle t_1}$} (b1); 
    \path[<->] (a1) edge[bend right = 60] node[yshift = -4pt] {${\scriptscriptstyle t_2}$} (c1);
     \path[<->] (c1) edge[bend right = 60] node[yshift = -4pt] {${\scriptscriptstyle t_3}$} (a4);
     \path[<->] (b1) edge[bend right = 60] node[xshift = 4.5pt] {${\scriptscriptstyle t_4}$} (a2);
\end{tikzpicture}
\caption{Lieb Lattice}
\label{DLieb}
\end{figure}
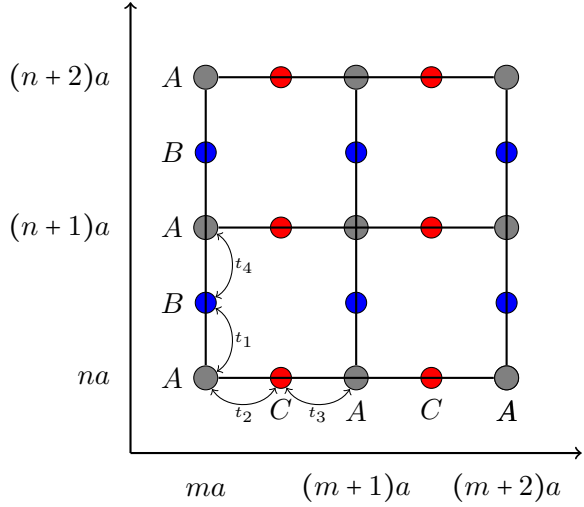

Similar to the conventional model in Sec. \ref{sec:model}, the two dimensional Hamiltonian can also be reduced to the following one dimensional operator after Fourier transform: $H_{\alpha,{\bf t},\theta}^g$  acting on $\ell^2(\Z; \C^3)$ as follows:
\begin{align}\label{def:H_alpha_theta_2}
    (H_{\alpha,{\bf t},\theta}^g u)^A_{m}&=e^{-\pi i m\alpha}(1+t_4 e^{2\pi i(\theta+m\alpha)})u^B_{m}+t_2 u^C_m+t_3 u^C_{m-1},\\
    (H_{\alpha,{\bf t},\theta}^g u)^B_m&=e^{\pi i m\alpha}(1+t_4 e^{-2\pi i(\theta+m\alpha)})u^A_m\\
    (H_{\alpha,{\bf t},\theta}^g u)^C_m&=t_2 u^A_m+t_3 u^A_{m+1}.
\end{align}
Similar to the connection to almost Mathieu operator for previous model, the related scalar model for the general hopping model in this section is:
\begin{align}
    u_{m+1}^A+u_{m-1}^A+(t_2 t_3)^{-1}(1+t_2^2+t_3^2+t_4^2+2t_4 \cos(2\pi(\theta+m\alpha))-E^2)u_m^A=0.
\end{align}
We state the theorems for the general hopping case and omit the proofs since they are identical to those of Theorems \ref{thm:Cantor_2d} and \ref{thm:spectral_dec}.
Let $\Sigma_{\alpha,{\bf t}}^g:=\sigma(H^g_{\alpha,{\bf t}})\setminus \{0\}$ be the {\it bulk of spectrum}.
\begin{theorem}\label{thm:Cantor_2d_general}
For $\alpha\notin \Z$, $E=0$ is a point spectrum of $H^g_{\alpha,{\bf t}}$, and $\mathrm{dist}(\Sigma^g_{\alpha,{\bf t}},\{0\})>0$. $H^g_{\alpha,{\bf t}}$ has purely continuous spectrum in $\Sigma^g_{\alpha,{\bf t}}$.
Furthermore, if $\alpha$ is irrational,
    \begin{itemize}
        \item $\Sigma^g_{\alpha,{\bf t}}$ is a Cantor set.
        \item for $t_4=t_2t_3$, $H^g_{\alpha,{\bf t}}$ has purely singular continuous spectrum in $\Sigma^g_{\alpha,{\bf t}}$, and the Hausdorff dimension of $\Sigma^g_{\alpha,{\bf t}}$ is at most $1/2$.
    \end{itemize}  
\end{theorem}
\begin{theorem}\label{thm:L_t<1}
For any $(t_2,t_3,t_4)\in \R_+^3$, any $\alpha\notin \Z$ and any $\theta$, $E=0$ is a point spectrum of $H^g_{\alpha,{\bf t},\theta}$, and  $\mathrm{dist}(\Sigma^g_{\alpha,{\bf t}}, \{0\})>0$.
\begin{itemize}
    \item For $0<t_2t_3<t_4$,
\begin{enumerate}
    \item for irrational $\alpha$ such that $0\leq \beta(\alpha)<\log(t_4/(t_2t_3))$ and $\theta$ such that $\gamma(\alpha,\theta)=0$, $H^g_{\alpha,{\bf t},\theta}$ exhibits Anderson localization in $\Sigma^g_{\alpha,{\bf t}}$.    
    \item for irrational $\alpha$ such that $\beta(\alpha)>\log (t_4/(t_2t_3))$ and all $\theta\in \T$, $H^g_{\alpha,{\bf t},\theta}$ has purely singular continuous spectrum in $\Sigma^g_{\alpha,{\bf t}}$.
    \item for irrational  $\alpha$ such that $\beta(\alpha)=0$ and $\theta$ such that $\gamma(\alpha,\theta)<\log (t_4/(t_2t_3))$, $H^g_{\alpha,{\bf t},\theta}$ exhibits Anderson localization in $\Sigma^g_{\alpha,{\bf t}}$.
    \item for irrational $\alpha$ such that $\beta(\alpha)=0$ and $\theta$ such that $\gamma(\alpha,\theta)>\log (t_4/(t_2t_3))$, $H^g_{\alpha,{\bf t},\theta}$ has purely singular continuous spectrum in $\Sigma^g_{\alpha,{\bf t}}$.
\end{enumerate}
\item For $t_2t_3>t_4>0$, for any irrational $\alpha$ and any $\theta\in\T$, $H^g_{\alpha,{\bf t},\theta}$ has absolutely continuous spectrum in $\Sigma^g_{\alpha,{\bf t}}$.
\item For $t_2t_3=t_4>0$, for any irrational $\alpha$ and any $\theta\in \T$, $H^g_{\alpha,{\bf t},\theta}$ has purely singular continuous spectrum in $\Sigma^g_{\alpha,{\bf t}}$.
\end{itemize}
\end{theorem}

\end{document}